\documentclass[twocolumn,superscriptaddress,pra]{revtex4-2}
 
\usepackage{amsmath, amssymb, amsthm}
\usepackage{mathtools}
\usepackage{xurl}
\usepackage{float}
\usepackage{graphicx} 
\usepackage{tikz}
\usepackage{caption}
\usepackage{pgfplots}
\pgfplotsset{compat=1.18}
\usetikzlibrary{shapes.geometric, arrows.meta, positioning, calc}
\usetikzlibrary{patterns,decorations.pathreplacing,calligraphy}
\usepackage{booktabs}
\usepackage{hyperref}
\hypersetup{
    colorlinks=true,
    linkcolor=blue,
    citecolor=blue,
    urlcolor=blue
}
\usepackage{algorithm}
\usepackage{algorithmic}
\newtheorem{theorem}{Theorem}[section]

\newtheorem{lemma}[theorem]{Lemma}

\newtheorem{assumption}[theorem]{Assumption}
\theoremstyle{definition}
\newtheorem{definition}[theorem]{Definition}

\allowdisplaybreaks

\begin{document}

\title{Imperfect Commitment in Maximal Extractable Value Auctions}
 
\author{Aleksei Adadurov}
\author{Sergey Barseghyan}
\author{Anton Chtepine}
\author{Antero Eloranta}
\author{Andrei Sebyakin}
\author{Arsenii Valitov}
\affiliation{nuconstruct \\ \textnormal{\texttt{team@nuconstruct.xyz}}}
 
\date{\today}
 
\begin{abstract}
Ethereum block builders run sealed auctions among searchers, but nothing in the protocol forces a builder to honor the auction outcome after observing submitted bundles. This paper studies the commitment problem. We model a builder who defects with probability $\varepsilon$ and, upon defection, replicates a type-specific fraction $\gamma(\tau)$ of the winning MEV opportunity. Searchers anticipate this behavior and choose between a risky first-price bid and a safe deterrence bid that makes frontrunning unprofitable. The resulting equilibrium is piecewise, with the cost of imperfect commitment depending jointly on replicability and competition. Using the \texttt{libmev} dataset, we estimate $\gamma(\tau)$ from right-tail bribe plateaus and decompose observed auction revenue against the surplus a defecting builder could capture. The results show sharp heterogeneity across MEV types: sandwich opportunities are already highly competitive, while naked arbitrage and liquidations leave substantially more surplus exposed to builder defection. Credible MEV auctions, therefore, require not only an auction format, but also constraints on the builder's ability to use observed bid and payload information ex post.
\end{abstract}

\maketitle

\newpage

\section{Introduction}

Block builders on Ethereum run auctions among searchers competing for inclusion in a block. The builder receives a set of bundles, each containing a sealed tip payment that constitutes the searcher's bid, selects winners, and submits the resulting block to the relay. This Maximal Extractable Value (or, later in the paper, MEV) auction is the empirically relevant clearing mechanism for searcher activity: it determines who captures price discrepancies, liquidations, sandwich attacks, and backrun opportunities, and at what cost. Nothing in the protocol forces the builder to honor the auction outcome. Having observed all submitted bundles and their transaction payloads, the builder can replicate the winning searcher's strategy himself, replace the searcher's transaction with his own, and capture the underlying MEV opportunity directly. This paper studies the commitment problem.

The relevant auction-theoretic benchmark is the corrupt-auctioneer literature \cite{compte2005corruption,lengwiler2010auctions,menezes2006corruption,burguet2004competitive}, in which the auctioneer's information advantage becomes a manipulation lever rather than a price-discovery tool. The auction-format question studied in \cite{adadurov2026open} is orthogonal: even a format-optimal sealed auction can be undermined by builder defection ex post. Under affiliated values, the same bid information that mitigates the winner's curse \cite{milgrom1982theory,krishna2010auction} also tells the builder which opportunities are worth replicating and at what cost, so the empirical force that pushes searchers to bid more aggressively is also what makes builder defection more profitable.

We summarize the builder's commitment by a parameter $\varepsilon \in [0,1]$: the share of MEV opportunities on which the builder defects from the honest auction outcome by exploiting observed bid information. At $\varepsilon=0$, the builder is fully committed, and the winning searcher pays his own bid, while at $\varepsilon=1$, the builder always defects, using the observed payload to replicate the winning strategy and capture a fraction $\gamma(\tau)$, where $\tau$ is the MEV type, of the underlying value. Intermediate values capture partial commitment. We treat $\varepsilon$ as a property of the builder and ask how it shapes equilibrium bids, searcher surplus, and observed auction revenue.

We model a builder who runs a sealed-bid first-price auction among $n$ searchers whose values are affiliated and type-specific, with MEV types $\tau \in \{\texttt{sandwich}, \texttt{naked arb}, \texttt{liquidation}, \texttt{backrun}\}$. With defection probability $\varepsilon$, the builder uses observed bid information to replicate a fraction $\gamma(\tau)$ of the winning opportunity whenever this value exceeds the winning bid. Note that the standard revert protection ensures that the searcher's transaction does not execute, and no auction payment is collected. Anticipating this, searchers choose between a risky first-price bid and a safe deterrence bid $b=\gamma(\tau)v$ that makes defection unprofitable. The resulting equilibrium is piecewise: below a type-specific cutoff, searchers bid as in the standard first-price auction, while above the cutoff, they bid enough to deter the builder. In addition, a complementary channel arises when some searchers are integrated with the builder, which converts defection into selective disclosure to a favored bidder; we mention this as a natural extension but do not analyze it in the baseline model.

Using the same \texttt{libmev} dataset as \cite{adadurov2026open}, we recover $\hat{\gamma}(\tau)$ from right-tail bribe plateaus, decompose observed auction revenue against the foregone surplus a defecting builder could capture, and compare bribe distributions across builders. The results suggest sharp heterogeneity across MEV types. Sandwiches are already highly competitive, with bribes near the full extracted value, so the additional value a builder could obtain by defecting is small. Naked arbitrage and liquidations leave much larger gaps between observed bids and estimated replication values, making them more exposed to builder manipulation. We then compare these estimates to the \cite{bergemann2022optimal} honest-disclosure benchmark, interpreting the honest-disclosure threshold as an upper bound on user-beneficial disclosure once strategic builder behavior is admitted.

The contribution is twofold. First, the paper isolates the builder-defection channel in MEV auctions from the auction-format question studied in \cite{adadurov2026open} and characterizes the symmetric equilibrium under partial commitment. Second, it shows that the welfare cost of imperfect commitment is sharply heterogeneous across MEV types, because the replicability parameter $\gamma (\tau)$ varies by strategy. A credible MEV auction is therefore not primarily a matter of choosing among auction formats: it is a matter of constraining the builder's ability to use observed bid information ex post. 

The remainder of the paper is organized as follows. Section \ref{sec:literature} places the analysis within three related literatures: optimal information disclosure, auctioneer corruption, and MEV market structure. Section \ref{sec:model} develops the model and characterizes the symmetric Bayesian Nash equilibrium under partial commitment. Section \ref{sec:data} takes the model to the \texttt{libmev} dataset, estimates $\hat{\gamma}(\tau)$ from right-tail bribe plateaus, and reports the revenue decomposition. Section \ref{sec:discussion} discusses design implications, and Section \ref{sec:conclusion} concludes.

\section{Related Literature}\label{sec:literature}

We organize the literature into three strands. The first treats information disclosure as a design variable chosen by an honest auctioneer; this gives the benchmark against which builder defection must be measured. The second is the auction-corruption tradition, which provides the primitives for an auctioneer who exploits observed bid information ex post. The third documents the institutional setting in which builder discretion operates: the MEV auction and the surrounding PBS market structure. We review each strand below and close with a synthesis of the gap.

\subsection{Optimal Information Disclosure as the Honest Benchmark}

A complementary literature characterizes the revenue-maximizing disclosure policy of an honest auctioneer. \cite{bergemann2022optimal} study the second-price auction and show that the optimal policy fully reveals low values while pooling high values above a critical quantile that depends only on the number of bidders. In our notation, this gives a natural benchmark $\varepsilon^{\mathrm{Berg}}(n)$: the maximum information release that an honest auctioneer would choose, and therefore an upper bound on the disclosure a defecting builder could plausibly justify as competition-enhancing rather than manipulation-enabling.

Related work shows why this benchmark should not be read as a single universal prescription. \cite{bergemann2007information,li2017discriminatory} show that optimal information structures can be asymmetric or type-contingent, while \cite{eso2007price} emphasize that an informed intermediary may be able to capture the informational rent. \cite{board2009revealing} adds an important caution for thin markets: when information changes the ranking of bidders' valuations, the allocation effect can make disclosure revenue-reducing with few bidders. These results imply that the honest-disclosure benchmark itself depends on type-specific competition, which directly motivates the type-by-type structure of our empirical analysis.

\subsection{Corrupt Auctioneers and Credible Auctions}

Our model is closest to the auction-corruption literature, which provides the relevant primitives for an auctioneer who exploits observed bid information ex post: selective readjustment opportunities \cite{compte2005corruption}, surplus sharing with a favored bidder \cite{lengwiler2010auctions}, distorted quality evaluation \cite{laffont1991auction,burguet2004competitive}, bidder-to-bidder collusion \cite{eso2004bribing} and post-bid renegotiation with the winner \cite{menezes2006corruption}. These mechanisms differ in whether they distort allocation, but they share a common feature: information that is valuable for price discovery can also become a lever for favoritism or side-payment extraction once the auctioneer can act on it.

Two lessons are especially relevant for MEV auctions. First, corruption effects are equilibrium effects, not merely transfers: anticipated corruption can sustain implicit collusion, change honest bidders' strategies, and alter participation incentives \cite{compte2005corruption,arozamena2009effect}. Second, the welfare loss depends on the specific channel. Revealing the standing high bid supports last-look or bid readjustment; revealing order contents supports replication-based frontrunning; revealing identities supports side deals with integrated searchers. Our baseline model isolates the frontrunning-by-replication channel, but the same $\varepsilon$ notation extends naturally to the other manipulation channels.

\subsection{MEV Auctions, Builder Market Structure and the Proposer-Builder Separation}

The MEV literature establishes the channels through which builder discretion translates into welfare costs. \cite{daian2020flash} shows how public transaction information creates priority-gas-auction races and consensus risk, while \cite{heimbach2022eliminating} formalizes the AMM sandwich channel through which order details such as size and slippage become directly exploitable. These papers establish replication-based frontrunning as the relevant manipulation channel in our setting: once the builder observes a winning bundle's payload, sandwich, and similar opportunities are mechanically reproducible, while liquidations and specialized arbitrage are reproducible only in part.

A separate literature fixes the market structure in which builder discretion operates. \cite{ma2025analysis} models the orderflow-auction layer under PBS and characterizes how builder advantages translate into equilibrium bids and centralization. \cite{yang2024decentralization,bahrani2024centralization} document and model builder centralization, which reduces the effective number of competing builders and concentrates discretion in a small number of agents. \cite{heimbach2024nonatomic,pai2023structural,oz2024who} further show that orderflow access, integrated searcher-builder structures, and non-atomic arbitrage create persistent informational asymmetries between builders and independent searchers, sharpening the corruption analog.

\subsection{Synthesis and Literature Gap}

Two observations follow from the literature above and motivate the model and empirical analysis developed in the remainder of the paper. First, the auction-corruption literature \cite{laffont1991auction,compte2005corruption,lengwiler2010auctions,menezes2006corruption,burguet2004competitive,eso2004bribing,arozamena2009effect} establishes that the welfare loss from auctioneer manipulation operates through specific channels -- side-payment extraction, selective disclosure to a favored bidder, post-bid readjustment, and bidder-to-bidder collusion -- and that its magnitude depends on the equilibrium response of honest bidders. The classical literature is, however, silent on which channel dominates in any specific institutional setting and on whether the magnitude of distortion varies across the objects being auctioned.

Second, the MEV literature \cite{daian2020flash,heimbach2022eliminating,heimbach2024nonatomic,pai2023structural} identifies the dominant channel in the builder layer as frontrunning by replication. Unlike pure side-payment models, replication-based defection is type-specific: $\gamma(\tau)$ is close to one for sandwich attacks and substantially below one for liquidations and specialized arbitrage. The market-structure literature \cite{yang2024decentralization,bahrani2024centralization,oz2024who} further shows that the effective number of competing builders is small, so anticipated defection plausibly distorts equilibrium bids at first order.

What the literature does not do is combine these strands. No existing paper simultaneously embeds (a) a builder-run sealed-bid MEV auction with affiliated, type-specific searcher values, (b) the auctioneer-corruption logic of \cite{compte2005corruption,menezes2006corruption}, and (c) the type-specific replicability parameter $\gamma(\tau)$ that emerges from the MEV institutional setting. The auction-format question studied in \cite{adadurov2026open} is the natural companion to the present analysis: that paper asks which format an honest builder should choose, while this paper asks what an opportunistic builder can extract regardless of the format chosen.

\section{The Model}\label{sec:model}


\subsection{The Environment}

We start our model with the primitives from \cite{adadurov2026open}. We consider $n$ risk-neutral searchers that compete for a single MEV opportunity of type $\tau$. As in \cite{adadurov2026open}, each searcher gets a latent signal $z_i = \sqrt{\rho}Z + \sqrt{1-\rho} u_i$, where $Z \sim \mathcal{N}(0,1)$ and 
$u_i \overset{iid}{\sim} \mathcal{N}(0,1)$, and values a particular opportunity at $v_i = \exp (\mu + \sigma z_i)$. We use an assumption on the distribution: $v_i \sim \text{LogNormal} (\mu, \sigma^2)$, which holds for all values of the affiliation parameter $\rho \in [0, 1)$, a parameter that captures the common information, such as victim transaction size or price discrepancies.

Additionally, we emphasize the differences between the MEV types $\tau \in \{\texttt{sandwich}, \texttt{naked arb},$ $ \texttt{liquidation}, \texttt{backrun}\}$ that are formalized through the effective bidder count $n(\tau)$ and the affiliation level $\rho(\tau)$. The previous empirical analysis has revealed that, for instance, sandwiches correspond to large $n$ and high $\rho$ (as they are publicly visible and commoditized), while liquidations correspond to small $n$ and moderate $\rho$ (as they are quite specialized and infrastructure-dependent).

\subsection{Builder-Run Auction}

In this setup, we suppose that a single builder operates a sealed first-price auction characterized by a defection parameter $\varepsilon \in [0,1]$. To be precise, $\varepsilon$ is the share of MEV opportunities on which the builder defects from the honest auction outcome by exploiting the bid vector ex post. In the baseline version of the model, we treat $\varepsilon$ as an exogenous variable, i.e., a property of the builder's commitment. The timing of the game, depicted in Figure \ref{fig:timing}, is as follows:

\begin{enumerate}
    \item The builder's defection rate $\varepsilon$ becomes known.
    \item Each of $N$ potential searchers decides whether to enter, and the number of entrants $n \le N$ is realized.
    \item Each entrant $i$ observes her own private signal $z_i$ and submits a (sealed) bid $b_i$ to the builder.
    \item Resolution of the game: with probability $\varepsilon$ the builder defects and uses the observed bids $(b_1, \dots, b_n)$ together with the winning transaction payload to frontrun the winner; with probability $1-\varepsilon$ the builder honors the auction outcome, i.e., the highest bidder wins and pays their own bid.
\end{enumerate}

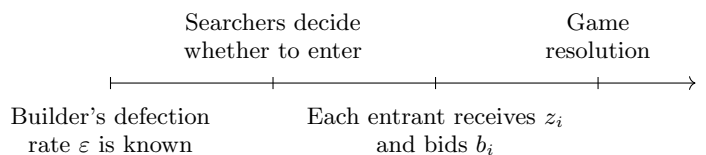
\begin{figure}[H]
\centering
\begin{tikzpicture}[x=2.15cm,y=1cm, every node/.style={font=\small}]
    \draw[->] (0,0) -- (3.6,0);

    \foreach \x in {0,1,2,3} {
        \draw (\x,0.08) -- (\x,-0.08);
    }

    \node[below=6pt, align=center] at (0,0) {Builder's defection \\ rate $\varepsilon$ is known};
    \node[above=6pt, align=center] at (1,0) {Searchers decide \\ whether to enter};
    \node[below=6pt, align=center] at (2,0) {Each entrant receives $z_i$ \\ and bids $b_i$};
    \node[above=6pt, align=center] at (3,0) {Game \\ resolution};

\end{tikzpicture}
\caption{\centering Timing of the game.}
\label{fig:timing}
\end{figure}

\textbf{Frontrunning.} In case of defection, the builder submits a competing transaction that captures a fraction $\gamma(\tau) \in [0,1]$ of the winner's MEV opportunity. Note that the parameter $\gamma$ is type-specific and reflects the replicability of the extraction strategy. For instance, for complex liquidations, $\gamma \ll 1$, as they require proprietary infrastructure, while sandwich attacks are highly replicable, $\gamma \approx 1$ since the victim's transaction is known.

\textbf{Revert protection.} In our setting, we consider that if the builder frontruns, the searcher's transaction reverts and no bid payment is collected. We integrate this notion into our model as it is a standard feature of bundle-based auction architectures (e.g., MEV-Share or Flashbots), which ensures that the frontrun replaces rather than supplements the searcher's payment. In addition, an important assumption is that the builder's decision upon observing bids is therefore straightforward: frontrun if and only if the frontrun value exceeds the winning bid, i.e., $\gamma (\tau) v_{(1)} > b_{(1)}$.

\subsection{Bidder Equilibrium}\label{sec:equilibrium}
First of all, it is important to make an emphasis on how the searchers act at the optimum, anticipating the builder's ex-post behavior. We characterize the symmetric Bayesian Nash equilibrium of the bidding stage, conditioning on $n$ entrants and based on the following assumptions.

\begin{assumption}[Extractability]
For any MEV type $\tau$, the builder possesses an ex-post frontrunning extraction capability $\gamma(\tau) \in [0,1]$. Furthermore, standard revert protection applies, meaning a frontrun searcher pays $0$.
\end{assumption}

\begin{assumption}[Affiliation]
Suppose that searchers' valuations are affiliated. Let $Y_1 = \max_{j \neq i} v_j$ denote the highest rival valuation. The conditional cumulative distribution function $H(y|v_i) = \mathbb{P}(Y_1 \le y | v_i)$ and conditional density $h(y|v_i)$ are well-defined and continuously differentiable.
\end{assumption}

\noindent Searcher $i$'s payoff with valuation $v_i$ and bid $b_i$ is characterized as follows:

\begin{equation}
    \pi(b_i, v_i) = \begin{cases}
        (1-\varepsilon) H(\beta^{-1}(b_i ) |v_i) (v_i - b_i) \quad \text{if} \ b_i < \gamma v_i \\
        H(\beta^{-1}(b_i ) |v_i) (v_i - b_i) \quad \text{if} \ b_i \ge \gamma v_i
    \end{cases}
\end{equation}

\noindent Note that in the first case, defection triggers frontrunning and the searcher's surplus is zeroed out, as only the no-defection event yields a positive payoff, while in the second case, the bid exceeds the frontrun value and the builder weakly prefers collecting the bid over frontrunning, which means that the searcher is paid regardless of defection.

\begin{lemma}
In the risky regime, i.e. for bids $b_i < \gamma v_i$, the symmetric equilibrium bid function $\beta(v_i)$ with boundary condition of $\beta(0) = 0$ is independent of $\varepsilon$ and satisfies the first-order ordinary differential equation:
\begin{equation}\label{eq:ODE}
    \beta'(v_i) = \frac{h(v_i | v_i)}{H(v_i | v_i)} \bigl(v_i - \beta(v_i) \bigr)
\end{equation}
\end{lemma}

\begin{proof}
    Note that the factor $1-\varepsilon$ only scales the payoff uniformly but does not change the first-order condition for the optimal bidding problem, as the searcher has the same maximizer of the two problems $\arg \max_{b_i} (1-\varepsilon) g(b_i) = \arg \max_{b_i} g(b_i)$. Consequently, the optimal bid in the risky regime is the standard affiliated first-price sealed-bid strategy. Conditional on winning, the expected utility is $\hat{\pi} (b_i |v_i) = (1-\varepsilon) H(\beta^{-1}(b_i)|v_i)(v_i - b_i)$. Differentiating with respect to $b_i$, setting $b_i = \beta(v_i)$ and equating to zero yields the standard affiliated hazard rate formulation.
\end{proof}

\begin{theorem}[Piecewise Nash Equilibrium]\label{thm:piecewise}
For any fixed defection rate $\varepsilon \in [0,1)$ and extraction capability $\gamma(\tau) \in (0,1]$, there exists a unique cutoff valuation $v^*(\varepsilon, \tau)$ such that the symmetric Bayesian Nash equilibrium bid function $b^*(v_i, \varepsilon, \tau)$ is given by:
\begin{equation}
    b^* (v_i, \varepsilon, \tau) = \begin{cases}
        \beta(v_i) \quad \text{if} \ v_i < v^* (\varepsilon, \tau) \\ 
        \gamma(\tau) v_i \quad \text{if} \ v_i \ge v^* (\varepsilon, \tau)
    \end{cases}
\end{equation}
where $\beta(v_i)$ is the solution to the ODE in Equation \ref{eq:ODE}.
\end{theorem}

\begin{proof}
    Speaking of the safe regime, for $b_i \ge \gamma v_i$, the cheapest bid that eliminates the frontrunning threat is $b_i = \gamma v_i$. One can see that any bid above $\gamma v_i$ is dominated, as it raises the payment without increasing the winning probability.

    The searcher compares the expected surplus in each regime. Conditional on winning, which is certain for $v = v_{(1)}$ in a symmetric equilibrium with a monotone bid function, the payoffs are as follows:

    \[\hat{\pi} (v_i) = (1-\varepsilon) (v_i - \beta(v_i))\]
    \[\tilde{\pi} (v_i) = (1-\gamma) v_i\]

    \noindent where $\hat{\pi} (\cdot)$ stands for the profit function in the risky regime, while $\tilde{\pi} (\cdot)$ stands for the profit function in the safe regime. Setting these expressions equal defines the indifference threshold:

    \begin{equation}\label{eq:indiff}
        \overline{\varepsilon} (v_i) = \frac{\gamma v_i - \beta(v_i)}{v_i - \beta (v_i)}
    \end{equation}

    \noindent Note that this expression is well-defined whenever $\gamma v_i > \beta (v_i)$, which holds precisely when the frontrun value exceeds the equilibrium bid. We assume that proportional shading falls with competition, i.e. $\frac{v_i - \beta(v_i)}{v_i}$ is decreasing with $v_i$, so $\overline{\varepsilon} (v_i)$ is strictly decreasing in $v_i$. Consequently, for any fixed $\varepsilon$, there exists a unique cutoff $v^*(\varepsilon, \tau)$, defined by $\overline{\varepsilon}(v^*) = \varepsilon$. Searchers with $v_i < v^*$ bid $\beta(v_i)$, while searchers with $v_i \ge v^*$ optimize by deviating to the limit price $\gamma(\tau) v_i$.
    
\end{proof}

\textbf{Remark.} In the risky regime, the bid function $\beta (\cdot)$ is derived from the affiliated first-price sealed-bid equilibrium of the full game. The truncation at $v^*_i$ modifies the distribution of competing bids, since the rivals with $v_{j\ne i} \ge v^*_i$ bid $\gamma v_j > \beta(v_j)$ and always outbid risky-regime bidders. A risky-regime bidder with $v_{j\ne i} < v^*_i$ therefore wins only if all rivals also have $v_{k \ne i,j} < v^*_i$. In principle, this changes the conditional order-statistic distribution entering the ODE. Here we argue that this effect is second-order when $v^*_i$ lies deep in the right tail of the log-normal distribution, which is the empirically relevant case, since with $\hat{\sigma} \approx 2.524$, the probability mass above $v^*_i$ is small for moderate $\varepsilon$.

\subsection{Builder's Ex-Ante Optimization}

We now turn to the builder's problem. The builder chooses $\varepsilon \in [0,1]$ to maximize total expected revenue, integrating over the unconditional density of the highest order statistic, $f_{(1)}(v)$.

\begin{definition}
    The expected revenue of the builder is given by the following function:
    \begin{equation}
    \begin{split}
        \mathbb{E}(R(\varepsilon)) = \int^{v^*(\varepsilon)}_0 \bigl[(1-\varepsilon) \beta(v) + \varepsilon \max \{ \beta(v) ,\gamma(\tau) v \} \bigr] \times \\ \times  f_{(1)}(v) \, dv + \int_{v^*(\varepsilon)}^\infty  \gamma(\tau) v  f_{(1)}(v) \, dv 
    \end{split}    
    \end{equation}
\end{definition}

\noindent Hence, given the builder's objective function, to solve for the optimal defection rate, we analyze the marginal revenue.

\begin{lemma}\label{thm:derivative}
    Assuming the frontrunning threat binds, i.e. $\gamma (\tau) v > \beta(v)$, such that $\max \{ \beta(v),\gamma(\tau) v \} = \gamma (\tau) v$, the derivative of the builder's expected revenue with respect to $\varepsilon$ is exactly given by:
    \begin{equation}
    \begin{split}
        \frac{d \mathbb{E}(R)}{d \varepsilon} = \int^{v^*(\varepsilon)}_0 \bigl[\gamma(\tau) v - \beta(v) \bigr] f_{(1)}(v) \, dv - \\- \frac{d v^*}{d \varepsilon} v^* \varepsilon \bigl(1- \gamma (\tau) \bigr) f_{(1)} (v^*)
    \end{split}
    \end{equation}
\end{lemma}

\begin{proof}
    Applying Leibniz's integral rule to the objective function yields the inner integral derivative plus a boundary term evaluated at $v^*$:
    \[\frac{d v^*}{d \varepsilon} \biggl[(1-\varepsilon) \beta(v^*) + \varepsilon \gamma(\tau) v^* - \gamma (\tau)v^* \biggr] f_{(1)}(v^*)\]

    \noindent From the indifference condition in Equation \ref{eq:indiff}, we know $(1-\varepsilon)(v^*-\beta (v^*)) = (1-\gamma(\tau))v^*$, which eventually rearranges to $(1-\varepsilon)\beta (v^*) = (\gamma(\tau) - \varepsilon) v^*$. Substituting this identity into the bracketed boundary gives $(\gamma(\tau) - \varepsilon) v^* + \varepsilon \gamma(\tau) v^* - \gamma (\tau) v^* = -v^* \varepsilon (1 - \gamma(\tau))$. Finally, this simplifies the full derivative to the stated form.
\end{proof}

\begin{theorem}[Optimal Builder Defection Rate]\label{thm:optimal}
    Assume the number of entrants $n$ is exogenous and let $\mathcal{V}$ denote the support of $v_{(1)}$. Then, the expected revenue maximization problem $\displaystyle \varepsilon^* = \arg \max_{0 \le \varepsilon \le 1} \mathbb{E}(R)$ admits strictly boundary solutions:
    \begin{enumerate}
        \item If $\gamma(\tau) v > \beta (v)$ for $\forall \ v \in \mathcal{V}$ (high extractability regime), then $\varepsilon^* = 1$.
        \item If $\gamma(\tau) v < \beta (v)$ for $\forall \ v \in \mathcal{V}$ (low extractability regime), then $\varepsilon^* = 0$.
    \end{enumerate}
\end{theorem}

\begin{proof}
    By Theorem \ref{thm:piecewise}, $\overline{\varepsilon} (v, \tau)$ is strictly decreasing in $v$, implying its inverse $v^* (\varepsilon)$ is strictly decreasing in $\varepsilon$. Hence, $\frac{d v^*}{d \varepsilon} < 0$. Because $\varepsilon$, $1-\gamma(\tau)$, and $f_{(1)}$ are all non-negative, the term $- \frac{d v^*}{d \varepsilon} v^* \varepsilon (1- \gamma (\tau)) f_{(1)} (v^*)$ is strictly positive.

    In the high extractability regime, the integrand in the first term, $\gamma (\tau) v - \beta (v)$, is strictly positive. Because the integral term is strictly positive and the boundary term is non-negative and strictly positive for $\varepsilon >0$, $\frac{d \mathbb{E} [R]}{d \varepsilon} > 0$ for $\forall \ \varepsilon \in [0, 1)$. This implies that expected revenue is strictly monotonically increasing, achieving its global supremum at the boundary $\varepsilon^* = 1$.

    Conversely, in the low extractability regime, $\max \{\beta (v), \gamma(\tau) v\} = \beta (v)$, hence, the safe regime integral disappears as $v^* \to \infty$. Thus, the equation trivially collapses to $\int^\infty_0 \beta (v) f_{(1)} (v) \, dv$, making the derivative with respect to $\varepsilon$ exactly $0$.
\end{proof}

\section{Data and Empirical Analysis}\label{sec:data}


\subsection{Data}
We use the same \texttt{libmev} dataset as in \cite{adadurov2026open}: 2.2 million MEV bundle transactions on Ethereum from September 2024 to August 2025, with a total extracted value of \$168.5M and total tips paid to builders of \$101.3M. Each record contains the transaction hash, block number, MEV type, builder, tip, and searcher's retained profit. We define extracted value as tip plus retained profit, so the observed bribe share $b_i/v_i$ is the empirical counterpart of the bid-to-value ratio in the model. Full variable definitions, the log-normal fit with $\hat\mu = 1.102$ and $\hat\sigma = 2.524$, and revenue concentration diagnostics are reported in \cite{adadurov2026open}.

Table \ref{tab:table1} reproduces the summary statistics. Three features are relevant for taking the model to the data. First, sandwich and naked-arbitrage transactions dominate the sample, providing enough observations to recover type-specific bribe schedules. Second, liquidations are rare but individually large, which places them close to the thin-market environment in which the gap between $\gamma(\tau)v$ and the competitive bid $\beta(v)$ can be economically large. Third, mean bribe shares vary substantially across types, indicating substantial variation in how much surplus is already absorbed by honest competition. In the model, this variation determines whether builder defection is profitable: when $b_i/v_i$ is already close to the replicable share $\gamma(\tau)$, the builder gains little from replacing the winner; when it is far below $\gamma(\tau)$, the defection channel is potentially large.

\begin{table}[H]
\footnotesize
\begin{ruledtabular}
\begin{tabular}{ccccccc}
Type & Count & Total & Mean & Median & Std Dev & Bribe \% \\
\hline
Sandwich & 890{,}967 & 66.5 & 74.6 & 3.01 & 1{,}842 & 95\% \\
Naked arb & 915{,}194 & 52.6 & 57.4 & 3.15 & 1{,}529 & 67\% \\
Backrun & 405{,}701 & 28.2 & 69.5 & 2.28 & 2{,}104 & 76\% \\
Liquidation & 4{,}759 & 21.2 & 4{,}462 & 157.3 & 38{,}716 & 68\% \\
\hline
All & 2{,}216{,}621 & 168.5 & 76.0 & 3.01 & 1{,}925 & 79\% \\
\end{tabular}
\end{ruledtabular}
\caption{\centering Summary statistics of extracted values. Total is reported in millions of USDC; mean, median, and
standard deviation is reported in USDC.}
\label{tab:table1}
\end{table}

These cross-type differences map onto the structural parameters $(n(\tau), \rho(\tau))$ from Section \ref{sec:model}: sandwiches are the empirical analogue of the low-extractability regime, while naked arbitrage and liquidations are the natural candidates for a binding frontrunning constraint. The builder field identifies five major builders: beaverbuild.org, Titan, BuilderNet (Beaver), bobTheBuilder, and BuilderNet (Flashbots), which together account for 93\% of transactions; cross-builder comparisons are reported in Appendix \ref{sec:appendix_builders}.

\subsection{Empirical analysis}\label{sec:empirical}

Before estimating the builder-defection channel, it is useful to check that the affiliated-values assumption is empirically plausible. Appendix Figure \ref{fig:affiliation_scatter} plots, within each block and MEV type, the log extracted value of the largest opportunity against the log extracted value of the second-largest opportunity. The relationship is positive in every category. This does not identify the full common-value structure, but it supports the maintained assumption that searchers' valuations are not independent draws.

The empirical analysis has three steps. First, we estimate the type-specific recoverable share $\gamma(\tau)$ from the right tail of the bribe schedule. Second, we compare observed bids with the counterfactual amount $\gamma(\tau)v_i$ that a defecting builder could recover by replicating the winning bundle. Third, we use the honest-disclosure benchmark from \cite{bergemann2022optimal} only as a reference point: it tells us how much information an honest revenue-maximizing seller might reveal, while our model asks how much of that same information becomes dangerous once the builder can defect.

\textbf{Bribe Schedule.} Each panel of Figure \ref{fig:piecewise} plots the average bribe share against log extracted value, using 50 quantile bins. Under the piecewise equilibrium in Theorem \ref{thm:piecewise}, low-value searchers remain in the risky regime and bid according to the standard first-price schedule $\beta(v_i)$. Above the cutoff $v^*(\varepsilon,\tau)$, searchers switch to the deterrence bid $b_i=\gamma(\tau)v_i$. The empirical signature is therefore a rising bribe schedule followed by a right-tail plateau, and the plateau provides an estimate of $\hat{\gamma}(\tau)$.

\vspace{1em}
\begin{center}
    \includegraphics[width=\linewidth]{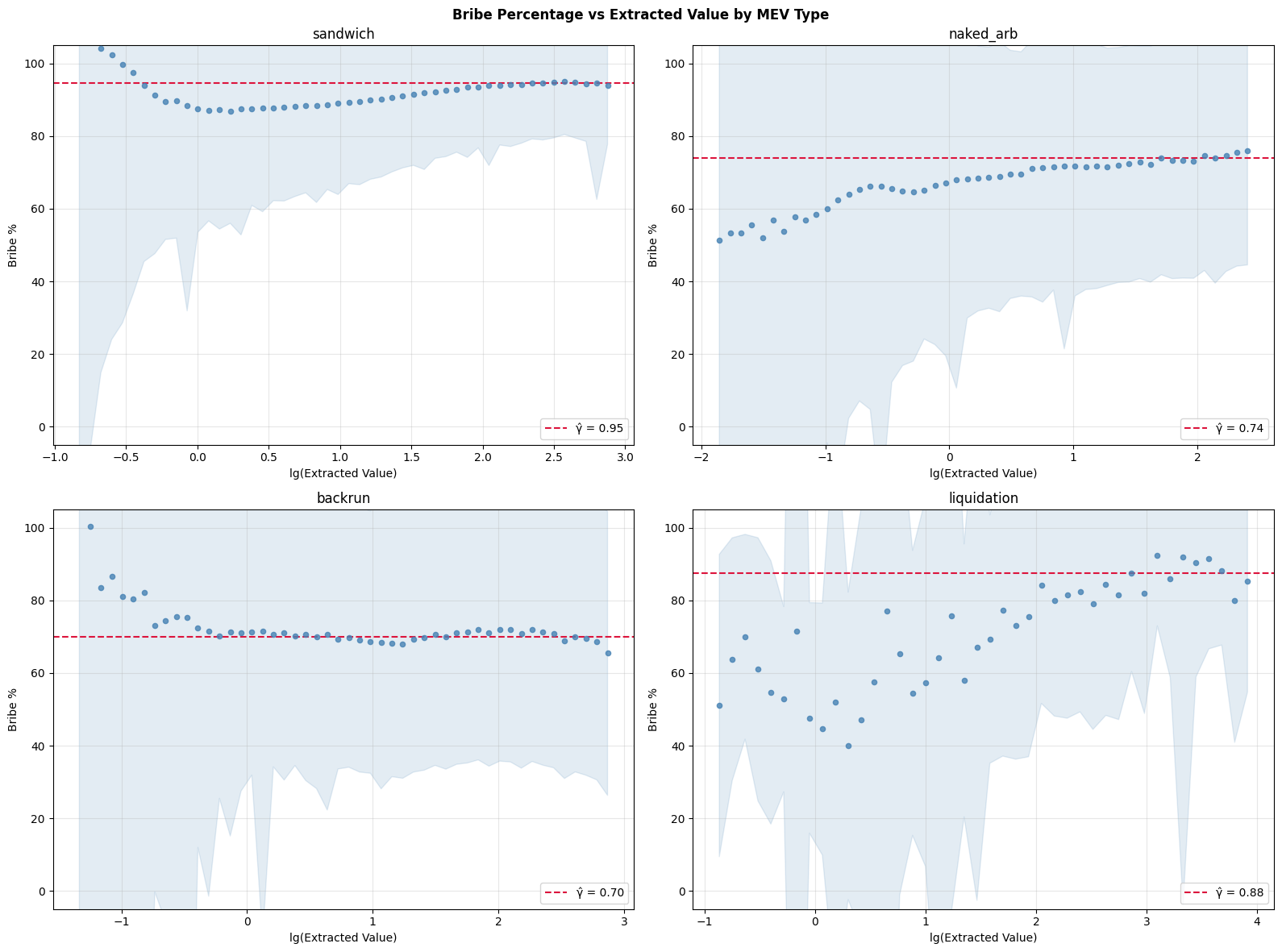}
    \captionof{figure}{\centering Bribe percentage vs extracted value by MEV type. The blue dots are bin means, the shaded band is plus or minus one standard deviation, and $\hat{\gamma}(\tau)$ is estimated from the right-tail plateau in the top 20\% of bins.}
    \label{fig:piecewise}
\end{center}
\vspace{1em}

For sandwiches, the schedule is nearly flat at 90--95\% across the value range. This is consistent with a large effective bidder count and a low residual surplus: the competitive bid $\beta(v_i)$ already absorbs most of the recoverable value, so the inequality $\gamma(\tau)v_i>\beta(v_i)$ rarely binds. There is therefore no sharp structural break to detect. In the language of Theorem~\ref{thm:optimal}, sandwiches are close to the low-extractability regime in which increasing the defection rate has little incremental value for the builder.

Naked arbitrage shows the clearest transition. The bribe share rises from roughly 45\% at low values to about 75\% in the right tail, with $\hat{\gamma}=0.74$. This is the pattern predicted by the piecewise equilibrium: proportional shading falls as valuations rise, and high-value searchers eventually bid close to the deterrence threshold. The economically relevant region is the gap between the low-value competitive bids and the right-tail plateau. In that region, a defecting builder can replace the winning searcher and recover approximately $\gamma(\tau)v_i$, which can be about 30 percentage points above the observed bid share.

Backruns are flatter, at roughly 70\%, with a slight decline in the highest-value bins. The flatness suggests more uniform competition across the value distribution, while the right-tail decline is not predicted by the baseline model and may reflect specialization among the highest-value opportunities. Liquidations are noisier because the sample is much smaller, but the schedule rises from roughly 55\% at low values to 85--90\% in the right tail. The estimate $\hat{\gamma}=0.88$ should therefore be read cautiously: it suggests that some liquidation opportunities are more replicable than expected, but the confidence bands also reflect a thin market with few effective bidders.

\textbf{Revenue Decomposition.} Figure \ref{fig:revenue} reports the empirical counterpart of the builder's defection gain. For each transaction, the potential gain is the positive part of $\hat{\gamma}(\tau)v_i-b_i$, aggregated by MEV type. This corresponds to the incremental revenue available when the builder defects rather than honors the auction outcome, and is the data analogue of the difference between the defection and no-defection terms in $\mathbb{E}(R(\varepsilon))$.

Across all types, the estimated foregone frontrun surplus is \$49.4M, or 48.8\% of observed tips. The largest contribution comes from naked arbitrage: \$24.3M of foregone surplus against \$18.3M of observed tips. This is the type where the model's defection channel is strongest, because $\hat{\gamma}=0.74$ exceeds the mean bribe share of 67\%, leaving a positive gap $\gamma(\tau)v_i-\beta(v_i)$ over a large part of the distribution. The magnitude is amplified by right-tail concentration; a small number of high-value transactions drive much of the \$24M. Appendix Figure \ref{fig:gini} reports the corresponding searcher-level concentration diagnostics.

\vspace{1em}
\begin{center}
    \includegraphics[width=\linewidth]{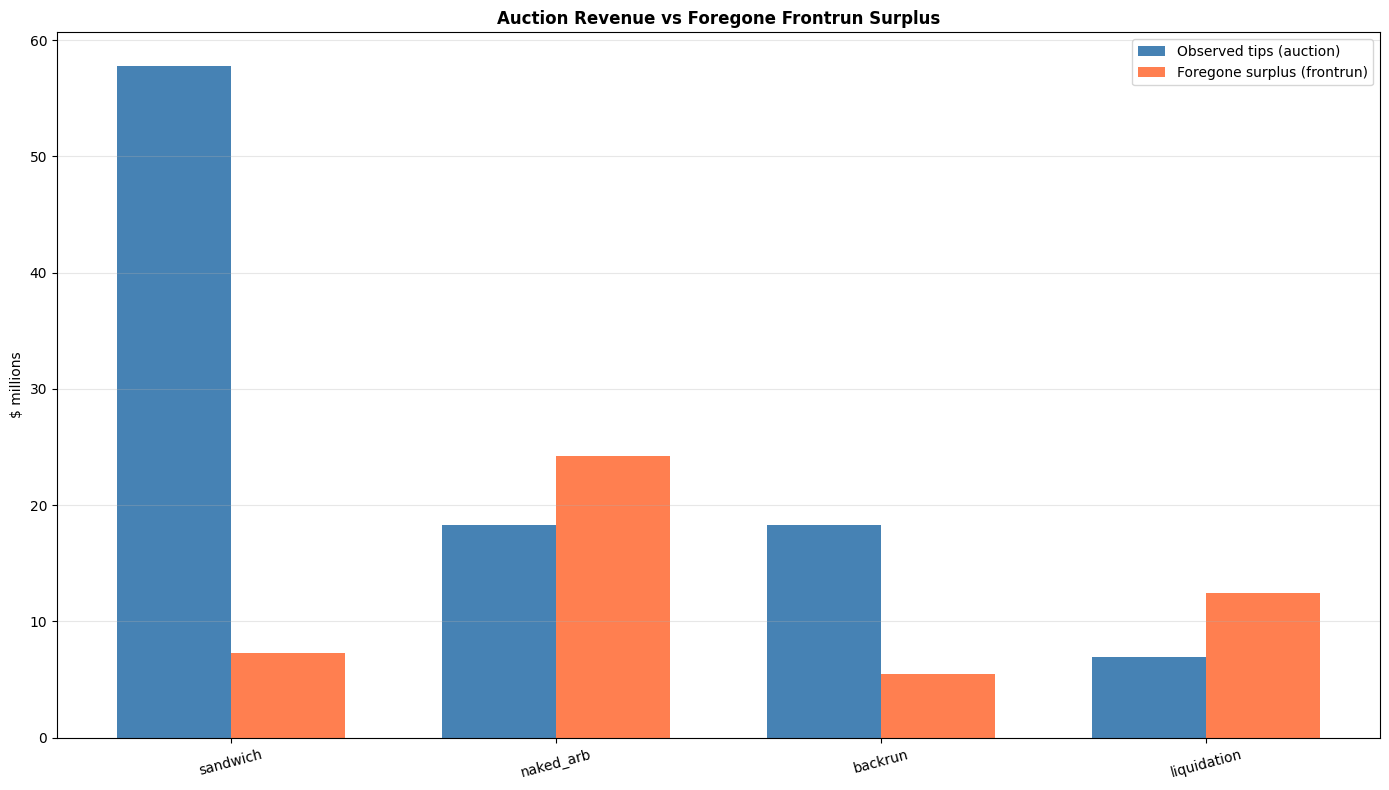}
    \captionof{figure}{\centering Auction revenue (blue) vs foregone frontrun surplus (orange) by MEV type.}
    \label{fig:revenue}
\end{center}
\vspace{1em}

Liquidations show the same mechanism on a smaller and noisier sample: \$12.4M foregone on \$7.0M of observed tips. The estimate $\hat{\gamma}=0.88$ is well above the 68\% mean bribe share, implying that the builder's recoverable value can be much larger than the competitive payment. Because there are only 4{,}759 liquidation observations, the average foregone surplus per liquidation is approximately \$2{,}600. These are precisely the individually large, thin-market opportunities for which the model predicts the largest commitment problem.

Sandwiches sit at the opposite extreme. The estimated foregone surplus is \$7.3M on \$57.8M of observed tips, because competitive bids already extract close to the full opportunity value. Most of the residual comes from the left tail, where very small sandwiches occasionally have bribe shares below the estimated plateau. Backruns are intermediate, with \$5.5M foregone on \$18.3M of observed tips.

\textbf{Bergemann Disclosure Threshold.} The previous exercise estimates the manipulation side of the model: how much surplus a builder could capture after observing bids and payloads. Figure \ref{fig:bergemann_threshold} adds the complementary benchmark from \cite{bergemann2022optimal}. In an honest second-price auction, the revenue-maximizing information policy reveals low values and pools the upper tail above a cutoff chosen so that only a small number of bidders remain effectively tied at the top. This gives a benchmark $\varepsilon^{\mathrm{Berg}}(n)$ that depends on the effective number of competing searchers.

\vspace{1em}
\begin{center}
    \includegraphics[width=\linewidth]{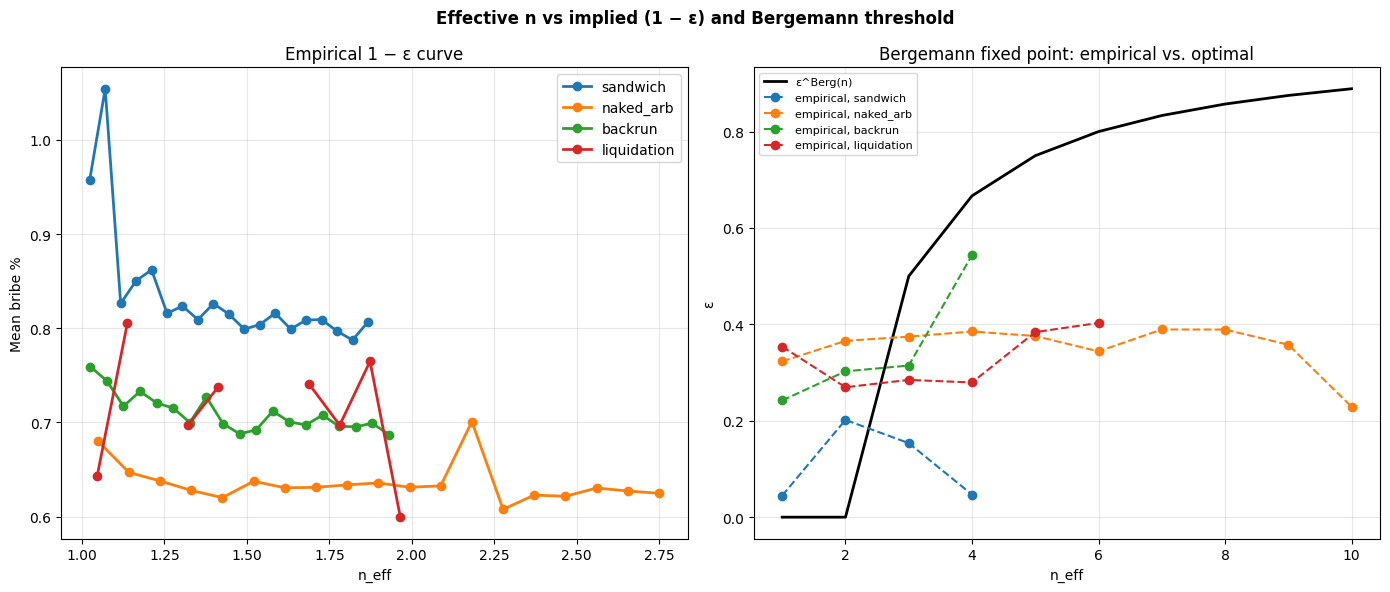}
    \captionof{figure}{\centering Bergemann disclosure benchmark by effective bidder count. The threshold is type-specific through $n(\tau)$: thicker markets can tolerate more information release, while thin markets require more pooling of the upper tail.}
    \label{fig:bergemann_threshold}
\end{center}
\vspace{1em}

The key implication is that the honest benchmark is type-specific because $n(\tau)$ is type-specific. Sandwiches have many effective bidders and already exhibit near-complete competitive extraction, so the honest-disclosure benchmark permits relatively more information release. Liquidations sit at the opposite end: the market is thin, specialized, and noisy, so the benchmark keeps the mechanism closer to pooling. Naked arbitrage and backruns fall between. This ordering is consistent with the revenue decomposition above. Appendix Figure \ref{fig:board} gives the related Board-style allocation-effect diagnostic by effective bidder count.

$\varepsilon^{\mathrm{Berg}}(n)$ should be interpreted as an upper bound on user-beneficial disclosure, not as the builder's optimal defection rate. \cite{bergemann2022optimal} study a seller who discloses information to improve auction revenue but does not privately exploit the signal. Once the builder can use observed bids and payloads to replicate the winner, the relevant object is the commitment parameter $\varepsilon$, and the admissible level of information exposure must be below the honest-disclosure benchmark whenever the defection gain is positive.

\section{Discussion and Design Implications}\label{sec:discussion}

The cost of imperfect builder commitment is type-specific. A single defection parameter $\varepsilon$ has very different welfare implications depending on $\gamma(\tau)$ and on how much of the opportunity value is already absorbed by competitive bidding.

This changes the design problem from choosing an auction format to enforcing credible commitment after bids are observed. A sealed first-price auction is not credible if the builder can inspect the payload, replace the winner, and avoid collecting the bid because the searcher's bundle reverts. The relevant safeguards are therefore those that reduce the usable information available to the builder at the moment of allocation, raise the cost of replication, or make defection observable and punishable. Aggregate or delayed information about competition may still preserve some linkage benefit, but real-time access to the winning payload, bidder identity, or standing high bid creates exactly the ex-post manipulation channel represented by $\varepsilon$.

\section{Conclusion}\label{sec:conclusion}

This paper studies a builder-run MEV auction in which the auctioneer cannot perfectly commit to honoring the auction outcome. After observing bids and payloads, the builder may replace the winning searcher's transaction and capture a type-specific fraction $\gamma(\tau)$ of the opportunity. The model summarizes this commitment problem with a defection parameter $\varepsilon$ and shows that searchers respond by choosing between a risky first-price bid and a safe deterrence bid.

The resulting equilibrium is piecewise. When the competitive bid $\beta(v)$ is high relative to the builder's recoverable value $\gamma(\tau)v$, defection has little value; when $\gamma(\tau)v$ exceeds the competitive bid, the builder's commitment problem binds, and high-value searchers must bid up to the deterrence threshold. The empirical analysis shows that this cost is sharply heterogeneous across MEV types.

The main design implication is that credible MEV auctions cannot be evaluated by auction format alone. The mechanism must also constrain the builder's ability to use observed bid information ex post. The appropriate commitment technology should depend on the MEV type, the effective number of bidders, and the replicability of the opportunity. Future work should estimate $\varepsilon$ and $\gamma(\tau)$ at the builder-searcher level, model selective disclosure to integrated searchers explicitly, and test mechanisms that preserve competition while withholding the payload and identity information that make replication profitable.

\onecolumngrid
\clearpage
\appendix 
\section{Additional Derivations}
\renewcommand{\thetable}{A\arabic{table}}
\setcounter{table}{0}
\renewcommand{\thefigure}{A\arabic{figure}}
\setcounter{figure}{0}

\subsection{Affiliation Diagnostics}

Figure \ref{fig:affiliation_scatter} provides the visual counterpart to the affiliation estimates used in Section \ref{sec:empirical}. Each panel compares the largest and second-largest log extracted values within the same block and MEV type; the positive slope is strongest for liquidations and naked arbitrage, but is present across all categories.

\begin{center}
    \includegraphics[width=\linewidth]{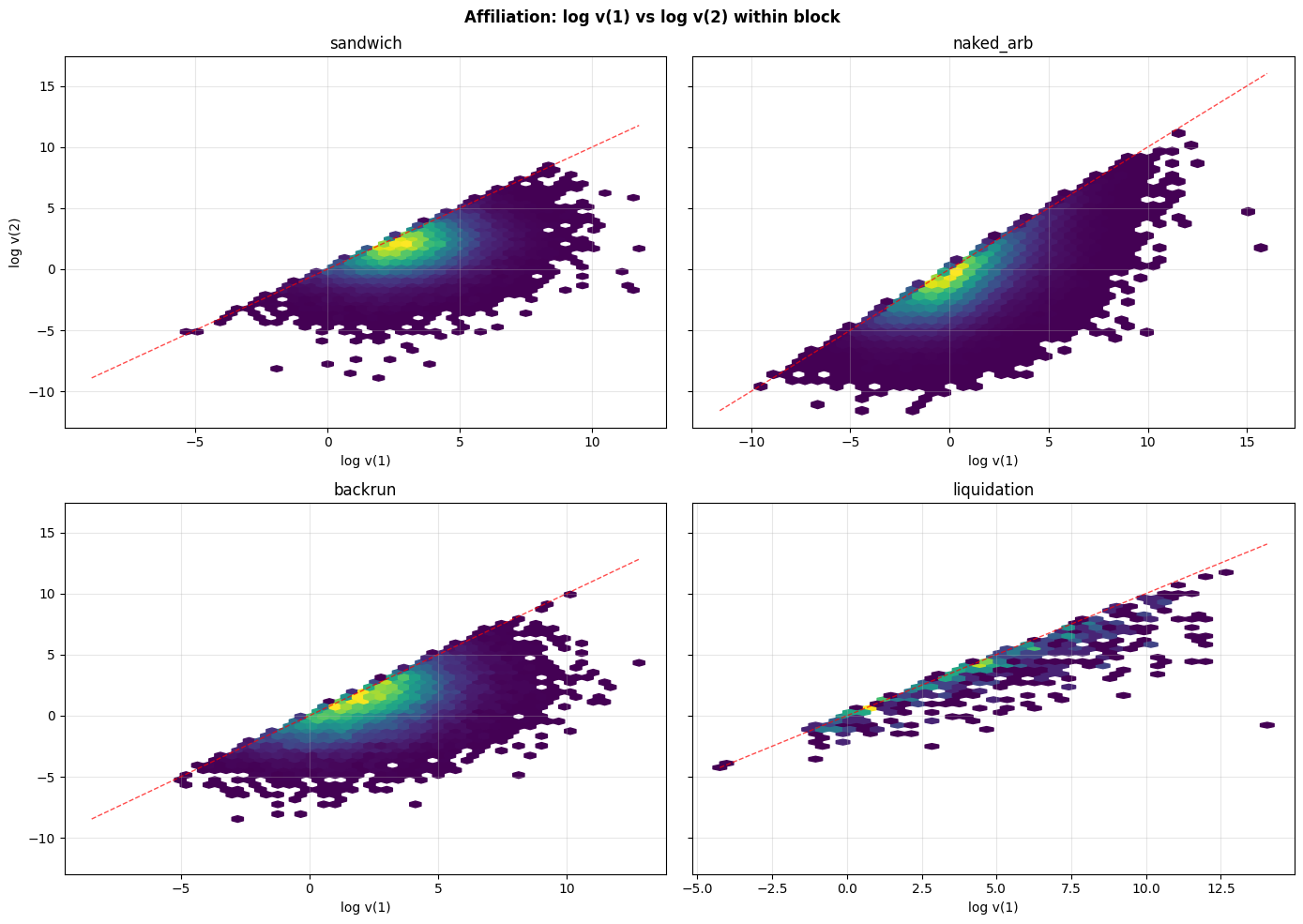}
    \captionof{figure}{\centering Affiliation diagnostic by MEV type. Each panel plots the largest against the second-largest log extracted value within a block.}
    \label{fig:affiliation_scatter}
\end{center}

\subsection{Searcher Concentration}

Figure \ref{fig:gini} shows that extracted MEV is highly concentrated at the searcher level across all types, which supports the right-tail interpretation of the revenue-decomposition results. 

\begin{center}
    \includegraphics[width=\linewidth]{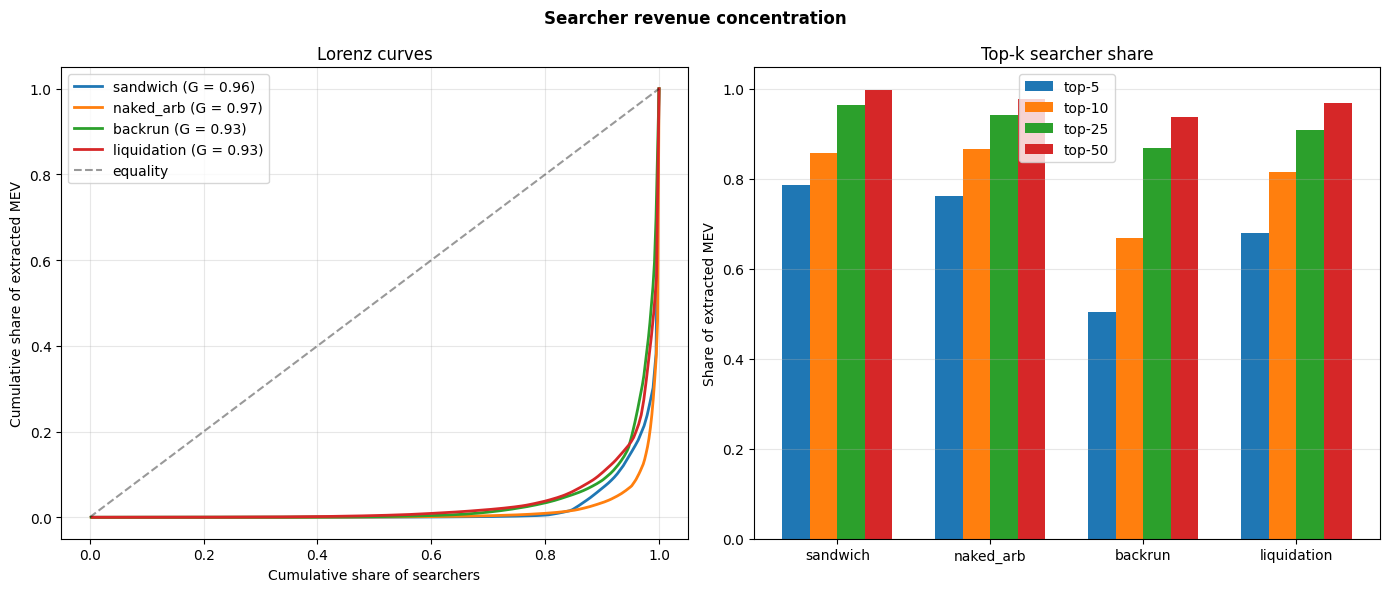}
    \captionof{figure}{\centering Searcher revenue concentration by MEV type. The Lorenz curves and top-$k$ shares show that extracted MEV is highly concentrated across searchers, with top searchers accounting for most extracted value in every category.}
    \label{fig:gini}
\end{center}

\subsection{Allocation-Effect Diagnostic}

Figure \ref{fig:board} reports the Board-style allocation-effect check: revenue and bribe shares do not move uniformly with effective bidder count. This supports the use of type-specific $n(\tau)$ in the honest-disclosure benchmark.

\begin{center}
    \includegraphics[width=\linewidth]{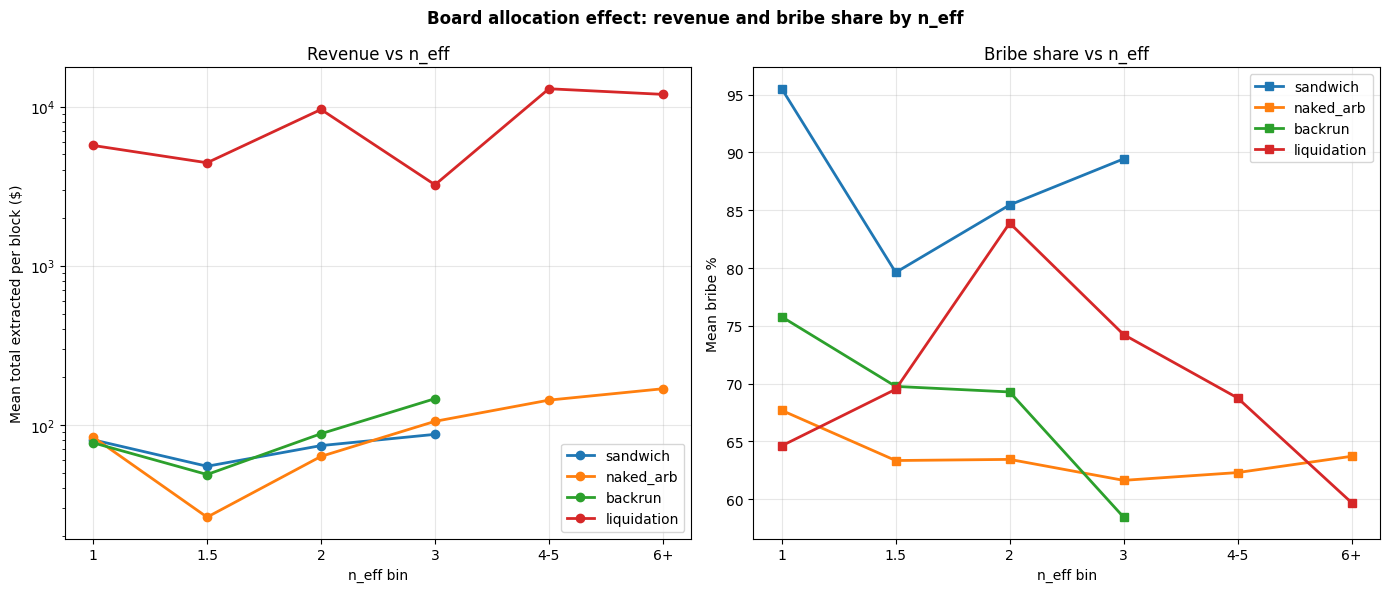}
    \captionof{figure}{\centering Board-style allocation-effect diagnostic. The panels compare revenue and bribe shares across effective-bidder-count bins, illustrating how the sign and strength of disclosure effects can vary with market thickness and MEV type.}
    \label{fig:board}
\end{center}

\subsection{Cross-Builder Bribe Distributions}\label{sec:appendix_builders}

Table \ref{tab:topbuilders} reports summary statistics for the eight largest builders by total extracted value. Mean bribe percentages are homogeneous across the top five builders, while Titan's bribe standard deviation is roughly $3\times$ beaverbuild's, consistent with the integrated searcher-builder extension of the model.

\begin{table}[H]
\begin{ruledtabular}
\begin{tabular}{cccccc}
Builder & Count & Total & Mean bribe & Bribe std & Number of searchers \\
\hline
beaverbuild.org         & 920,935 & 59.50 & 0.810 &  3.505 & 529 \\
Titan (titanbuilder.xyz)& 976,611 & 40.35 & 0.802 & 10.101 & 558 \\
BuilderNet (Beaver)     &  57,981 & 24.76 & 0.778 &  1.759 & 173 \\
bobTheBuilder.xyz       &   5,134 & 10.27 & 0.795 &  0.251 &  68 \\
BuilderNet (Flashbots)  &  64,360 &  7.58 & 0.781 &  2.769 & 184 \\
Ty For The Block        &   3,419 &  5.94 & 1.008 &  0.180 &   9 \\
(unnamed)               &   3,703 &  4.87 & 0.874 &  0.288 &  24 \\
BuildAI (buildai.net)   &   2,134 &  4.62 & 0.833 &  0.314 &  63 \\
\end{tabular}
\end{ruledtabular}
\caption{\centering Top builders by total extracted value. Total is reported in millions of USDC; mean bribe and bribe std are reported as fractions of extracted value.}
\label{tab:topbuilders}
\end{table}

Figure \ref{fig:cross_builder} disaggregates by MEV type. Sandwich bribes cluster at 95--100\% regardless of builder, while naked arb and liquidation exhibit larger cross-builder dispersion, consistent with the type-specific $\gamma(\tau)$ structure.

\begin{center}
    \includegraphics[width=\linewidth]{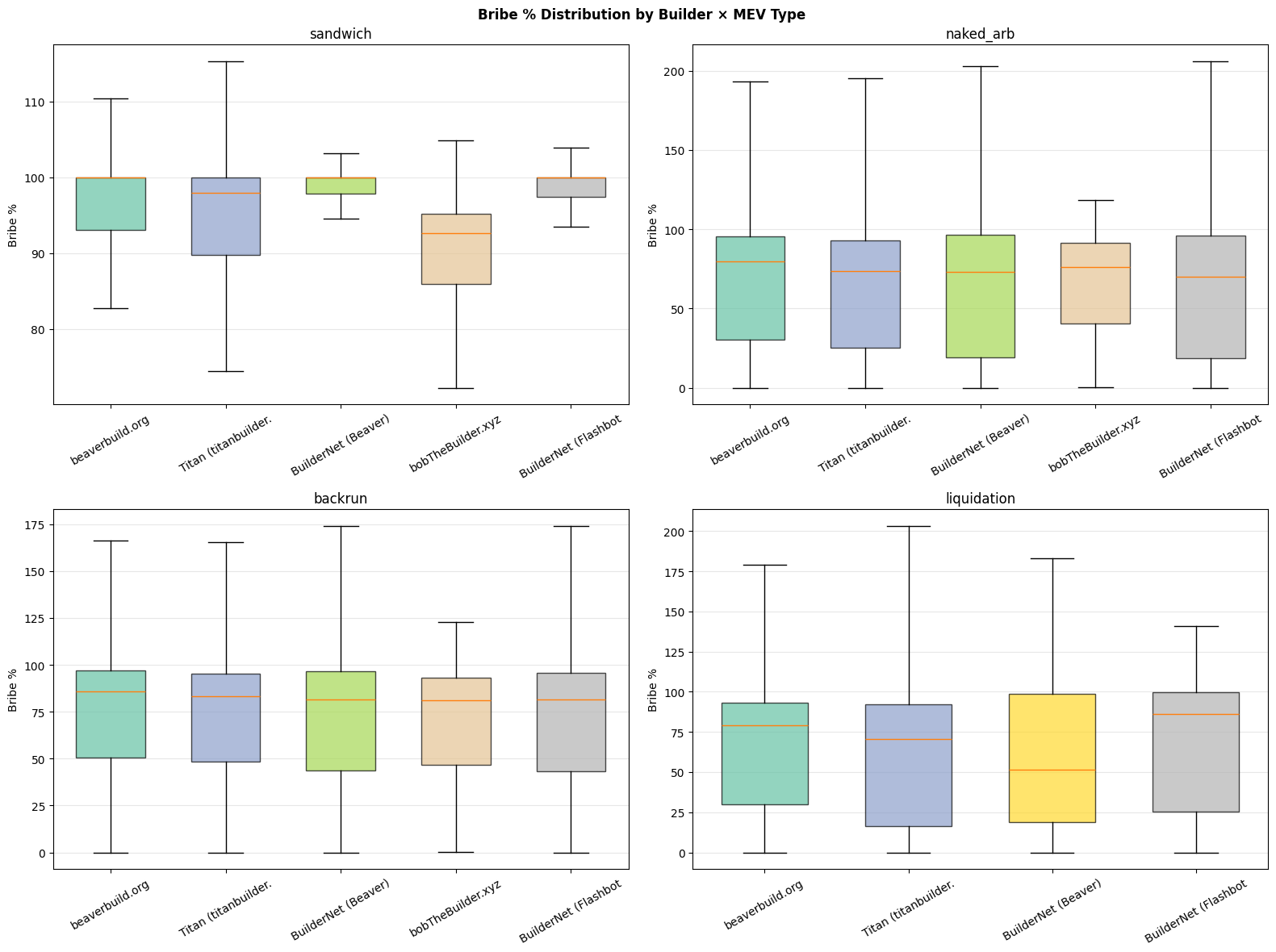}
    \captionof{figure}{\centering Bribe percentage distribution by builder $\times$ MEV type. Boxplots show median, interquartile range, and whiskers at 1.5$\times$ interquartile range; outliers suppressed. The top five builders by total extracted value are shown.}
    \label{fig:cross_builder}
\end{center}

\end{document}